\documentclass{llncs}

\usepackage{color}
\definecolor{lcol}{rgb}{0,0,0.5}

\usepackage[pdftex,hyperindex=false,bookmarks=false,
	        colorlinks=true,linkcolor=lcol,citecolor=lcol,
            filecolor=lcol,pagecolor=lcol,urlcolor=lcol]{hyperref}

\usepackage[capitalize]{cleveref}

\usepackage{amsmath,amssymb}
\usepackage{mathpartir}

\usepackage{fancyhdr}
\title{An Under-Approximate Relational Logic}
\subtitle{Heralding Logics of Insecurity, Incorrect Implementation \& More}

\author{Toby Murray}
\institute{University of Melbourne and Data61\\\email{toby.murray@Unimelb.edu.au}}

\begin{document}

\maketitle

\thispagestyle{fancyplain}

\begin{abstract}
  Recently, authors have proposed \emph{under-approximate} logics for
  reasoning about programs. So far, all such logics have been confined to
  reasoning about individual program behaviours. Yet there exist many
  over-approximate \emph{relational} logics for reasoning about pairs of
  programs and relating their behaviours.

  We present the first under-approximate relational
  logic, for the simple imperative language IMP.
  We prove our logic
  is both sound and complete.
  Additionally,
  we show how reasoning in this logic
  can be decomposed into non-relational reasoning in an under-approximate
  Hoare logic, mirroring Beringer's result for over-approximate
  relational logics.
  We illustrate the application of our logic on some small
  examples in which we provably demonstrate the presence of insecurity.  
\end{abstract}

\section{Introduction}
  Almost all program logics are \emph{over-approximate}, in that they
  reason about an over-approximation of program behaviour. Recently,
  authors have proposed 
  \emph{under-approximate} logics~\cite{OHearn_19,deVries_Koutavas_11},
  for their promise to increase the
  scalability of reasoning methods, avoid false positives, and
  \emph{provably} detect bugs.

  So far, all such logics have been confined to reasoning about
  individual program executions. Yet, many properties of interest 
  talk about pairs of program behaviours. Examples are security
  properties like noninterference~\cite{Goguen_Meseguer_82}, but also
  function sensitivity,
  determinism, correct implementation (refinement)
  and program equivalence.
  Such properties can be reasoned about using \emph{relational} logics~\cite{Benton_04},
  in which pre- and post-conditions are relations on program states.
  Over-approximate relational logics for doing so have been
  well-studied~\cite{Benton_04,barthe2017proving,Ernst_Murray_19,maillard2019next}.

  Here, we present (\cref{sec:logic}) the first under-approximate relational
  logic, for the simple imperative language IMP~\cite{Nipkow_Klein:Isabelle}.
  We prove our logic
  is both sound and complete (\cref{sec:soundness-completeness}).
  Additionally (\cref{sec:decomposition}),
  we show how reasoning in this logic
  can be decomposed into non-relational reasoning in an under-approximate
  Hoare logic, mirroring Beringer's result for over-approximate
  relational logics~\cite{Beringer_11}.
  Under-approximate logics are powerful for provably demonstrating the
  presence of bugs.
  We illustrate the application of our logic on some small
  examples (\cref{sec:using}) for provably demonstrating violations of
  noninterference security. In this sense, our logic can be seen as the
  first \emph{insecurity logic}. However, being a general relational logic,
  it should be equally applicable for provably demonstrating violations
  of the other aforementioned relational properties.
  All results in this paper have been mechanised in Isabelle/HOL~\cite{Nipkow_PW:Isabelle}.

\section{The Logic}
\label{sec:logic}

\subsection{The Language IMP}

\newcommand{\kw}[1]{\textbf{#1}}
\newcommand{\Skip}{\kw{skip}}
\newcommand{\ITE}[3]{\kw{if}\ #1\ \kw{then}\ #2\ \kw{else}\ #3}
\newcommand{\While}[2]{\kw{while}\ #1\ \kw{do}\ #2}
\newcommand{\Seq}[2]{#1;\ #2}
\newcommand{\Assign}[2]{#1 := #2}

Our logic is defined for the simple imperative language~IMP~\cite{Nipkow_Klein:Isabelle}. Commands~$c$ in IMP include variable assignment, sequencing,
conditionals (if-then-else statements), iteration (while-loops), and
the no-op \Skip.

\[
c ::=  \Assign{x}{e}\ |\ \Seq{c}{c}\ |\ \ITE{b}{c}{c}\ |\ \While{b}{c}\ |\ \Skip
\]
Here, $e$ and~$b$ denote respectively (integer) arithmetic and boolean
expressions.

\newcommand{\sem}[3]{(#1,#2) \Rightarrow #3}
The big-step semantics of IMP is entirely standard (omitted for brevity)
and defines its semantics over states~$s$ that map variables~$x$ to (integer)
values.
It is characterised by the judgement $\sem{c}{s}{s'}$ which denotes that
command~$c$ when execute in initial state~$s$ terminates in state~$s'$.

\subsection{Under-Approximate Relational Validity}

\newcommand{\rel}[3]{#1(#2,#3)}

Relations~$R$ are binary relations on states. We write~$\rel{R}{s}{s'}$ when
states~$s$ and~$s'$ are related by~$R$.

\newcommand{\valid}[4]{\vDash \langle#1\rangle\ #2,\ #3\ \langle#4\rangle}

Given pre- and post-relations~$R$ and~$S$ and commands~$c$ and~$c'$,
under-approximate relational validity is denoted $\valid{R}{c}{c'}{S}$ and
is defined
analogously to its Hoare logic counterpart~\cite{deVries_Koutavas_11,OHearn_19}.

\begin{definition}[Under-Approximate Relational Validity]\label{defn:valid}
\[
\valid{R}{c}{c'}{S} \equiv (\forall\ t\ t'.\ \rel{S}{t}{t'} \implies (\exists\ s\ s'.\ \rel{R}{s}{s'} \land \sem{c}{s}{t} \land \sem{c'}{s'}{t'}))
\]
\end{definition}
It requires that for all final states~$t$ and~$t'$ in~$S$, there exists
pre-states~$s$ and~$s'$ in~$R$ for which executing~$c$ and~$c'$ from each
respectively
terminates in~$t$ and~$t'$.

\subsection{Proof Rules}

\newcommand{\prove}[4]{\vdash \langle#1\rangle\ #2,\ #3\ \langle#4\rangle}

Judgements of the logic are written~$\prove{R}{c}{c'}{S}$. \cref{fig:rules}
depicts the rules of the logic. For $\prove{R}{c}{c'}{S}$ we define a set
of rules for each command~$c$. Almost all of these rules has a counterpart
in under-approximate Hoare logic~\cite{OHearn_19,deVries_Koutavas_11}.
There are also
three additional rules (\textsc{Conseq}, \textsc{Disj}, and \textsc{Sym}).
Of these, only the final one~\textsc{Sym} does not have a counterpart in prior
under-approximate logics.

The \textsc{Skip} rule simply encodes the
meaning of $\valid{R}{\Skip}{c'}{S}$ in the premise. The \textsc{Seq} rule
essentially encodes the intuition that $\valid{R}{(\Seq{c}{d})}{c'}{S}$ holds
precisely when $\valid{R}{(\Seq{c}{d})}{(\Seq{c'}{\Skip})}{S}$ does.
However we can also view this rule differently: once $c_1$ and~$c'$ have
both been executed, the left-hand side still has~$c_2$ to execute; however
the right-hand side is done, represented by~$\Skip$. The relation~$S$ then
serves as a witness for linking the two steps together. 

\newcommand{\AssignPost}[3]{#1[#2:=#3]}

This same pattern is seen in the other rules too: given the pre-relation~$R$
the \textsc{Assign} rule
essentially computes the witness relation~$\AssignPost{R}{x}{e}$ that holds
following the assignment~$\Assign{x}{e}$, which (like the corresponding
rule in~\cite{OHearn_19}) simply applies the strongest-postcondition
predicate transformer to~$R$.

The rules for if-statements, \textsc{IfTrue} and~\textsc{IfFalse}, are the relational analogues of the corresponding
rules from under-approximate Hoare logic. The rules~\textsc{WhileTrue}
and~\textsc{WhileFalse} encode the branching structure of unfolding a
while-loop. The rule \textsc{Back-Var} is the analogue of the
``backwards variant'' iteration rule from~\cite{OHearn_19,deVries_Koutavas_11}. Here
$R$ is a family of relations indexed by natural number~$n$ (which
can be thought of as an iteration count) such that the corresponding
relation~$R(k)$ must hold after the~$k$th iteration.
Unlike the
corresponding rule in~\cite{OHearn_19}, this rule explicitly encodes
that the loop executes for some number of times (possibly 0) to reach a frontier,
at which point $\exists\ n.\ R(n)$ holds and the remainder of the loop
is related to command~$c'$. 

\newcommand{\flip}[1]{\stackrel{\leftrightarrow}{#1}}

The \textsc{Conseq} and \text{Disj} rules are the relational analogues
of their counterparts in~\cite{OHearn_19}. Finally, the \textsc{Sym} rule
allows to flip the order of~$c$ and~$c'$ when reasoning about
$\prove{R}{c}{c'}{S}$ so long as the relations are flipped accordingly,
denoted $\flip{R}$ etc. Naturally, this rule is necessary for the logic to
be complete.

\newcommand{\eval}[2]{\lbrack#1\rbrack_{#2}}
\newcommand{\one}[1]{#1^1}

\begin{figure}
  \begin{mathpar}
    \infer{\forall\ t\ t'.\ \rel{S}{t}{t'} \implies \exists s'.\ \rel{R}{t}{s'} \land \sem{c'}{s'}{t'}}
          {\prove{R}{\Skip}{c'}{S}}
          \textsc{Skip}
          
    \infer{\prove{R}{c}{c'}{S} \and \prove{S}{d}{\Skip}{T}}
          {\prove{R}{(\Seq{c}{d})}{c'}{T}}
          \textsc{Seq1}

    \infer{\prove{\AssignPost{R}{x}{e}}{\Skip}{c'}{S}}
          {\prove{R}{\Assign{x}{e}}{c'}{S}}
          \textsc{Assign}

    \infer{\prove{R \land \one{b}}{c_1}{c'}{S}} 
          {\prove{R}{\ITE{b}{c_1}{c_2}}{c'}{S}}
          \textsc{IfTrue}
          
    \infer{\prove{R \land \one{\lnot b}}{c_2}{c'}{S}}
          {\prove{R}{\ITE{b}{c_1}{c_2}}{c'}{S}}
          \textsc{IfFalse}
          
    \infer{\prove{R \land \one{\lnot b}}{\Skip}{c'}{S}}
          {\prove{R}{\While{b}{c}}{c'}{S}}
          \textsc{WhileFalse}

    \infer{\prove{R \land \one{b}}{(\Seq{c}{\While{b}{c}})}{c'}{S}}
          {\prove{R}{\While{b}{c}}{c'}{S}}
          \textsc{WhileTrue}

    \infer{\forall\ n. \prove{R(n) \land \one{b}}{c}{\Skip}{R(n+1)} \\ 
           \prove{\exists\ n.\ R(n)}{\While{b}{c}}{c'}{S}}
          {\prove{R(0)}{\While{b}{c}}{c'}{S}}
          \textsc{Back-Var}

    \infer{R_1 \implies R_2 \and \prove{R_1}{c}{c'}{S_1} \and S_2 \implies S_1}
          {\prove{R_2}{c}{c'}{S_2}}
          \textsc{Conseq}

    \infer{\prove{R_1}{c}{c'}{S_1} \and \prove{R_2}{c}{c'}{S_2}}
          {\prove{R_1 \lor R_2}{c}{c'}{S_1 \lor S_2}}
          \textsc{Disj}

    \infer{\prove{\flip{R}}{c'}{c}{\flip{S}}}
          {\prove{R}{c}{c'}{S}}
          \textsc{Sym}
  \end{mathpar}
  \caption{Proof rules.\label{fig:rules} \AssignPost{R}{x}{e} denotes the
    transformer on relation~$R$ following the assignment~$\Assign{x}{e}$. It
  is defined as: $\AssignPost{R}{x}{e} \equiv \lambda\ s\ s'.\ \exists\ v.\ \rel{R}{s(x:=v)}{s'} \land s(x) = \eval{e}{s(x:=v)}$. Logical operators $\land$ and $\lor$ are lifted to relations, e.g.: $R \land S \equiv \lambda\ s\ s'.\ \rel{R}{s}{s'} \land \rel{S}{s}{s'}$. For boolean expression~$b$, $\one{b} \equiv \lambda\ s\ s'.\ \eval{b}{s}$. Also, $\exists\ n.\ R(n) \equiv \lambda\ s\ s'.\ \exists\ n.\ \rel{R(n)}{s}{s'}$. $R \implies S$ denotes when relation~$R$ is a subset of relation~$S$. Finally, $\flip{R}\ \equiv \lambda\ s\ s'.\ \rel{R}{s'}{s}$.}
\end{figure}

\section{Soundness and Completeness}
\label{sec:soundness-completeness}

\begin{theorem}[Soundness]\label{thm:soundness}
  For relations~$R$ and~$S$ and commands~$c$ and~$c'$, if
  $\prove{R}{c}{c'}{S}$ then $\valid{R}{c}{c'}{S}$.
\end{theorem}
\begin{proof}
  By structural induction on~$\prove{R}{c}{c'}{S}$. Almost all cases
  are straightforward. The only exception is the \textsc{Back-Var}
  rule, which, similarly to~\cite{OHearn_19}, requires an induction on
  the argument~$n$ to~$R(n)$ and generalising the result from~$R(0)$ to
  some~$R(k)$ to get a sufficiently strong induction hypothesis.
  \end{proof}

\begin{theorem}[Completeness]\label{thm:completeness}
  For relations~$R$ and~$S$ and commands~$c$ and~$c'$, if
  $\valid{R}{c}{c'}{S}$ then $\prove{R}{c}{c'}{S}$.
\end{theorem}
\begin{proof}
  By induction on~$c$. This proof follows a similar structure to the
  completeness proof of~\cite{OHearn_19}. 
\end{proof}

\section{Under-Approximate Relational Decomposition}
\label{sec:decomposition}

For over-approximate relational logic, Beringer~\cite{Beringer_11}
showed how one can decompose relational reasoning into
over-approximate Hoare logic reasoning.
We show that the same is true for under-approximate relational logic,
in which relational reasoning can be decomposed into
\emph{under-approximate} Hoare logic reasoning.

\newcommand{\decomp}[4]{\mathsf{decomp}(#1,#2,#3,#4)}
The essence of the approach is to decompose
$\valid{R}{c}{c'}{S}$ into two statements: one about~$c$ and the other
about~$c'$. $\valid{R}{c}{c'}{S}$ says that for every $S$-pair of
final states $t$ and~$t'$, we can find an~$R$-pair of states~$s$ and~$s'$
from which~$t$ and~$t'$ can be reached respectively by executing~$c$
and~$c'$. To decompose this we think about what it would mean to first
execute~$c$ on its own. At this point we would have reached a state~$t$
which has a $c$-predecessor~$s$ for which $\rel{R}{s}{s'}$, and the
execution of~$c'$ would not yet have begun and would still be in the
state~$s'$ which has a $c'$-successor~$t'$ for which~$\rel{S}{t}{t'}$.
We write $\decomp{R}{c}{c'}{S}$ that defines this relation that holds between~$t$ and~$s'$:
\[
\rel{\decomp{R}{c}{c'}{S}}{t}{s'} \equiv \exists\ s\ t'.\ \rel{P}{s}{s'} \land \sem{c}{s}{t} \land \sem{c'}{s'}{t'} \land \rel{Q}{t}{t'}
\]

This allows us to decompose $\valid{R}{c}{c'}{S}$ into two separate
relational statements about~$c$ and~$c'$. Here we use~$\Skip$ to indicate
when one command executes but the other does not make a corresponding step,
to simulate the idea that first~$c$ executes followed by~$c'$.

\begin{lemma}[Relational Validity Decomposition]\label{lem:decomp}
For relations~$R$ and~$S$ and commands~$c$ and~$c'$, 
\[
\valid{R}{c}{c'}{S} \iff
\left(\begin{array}{l}
\valid{R}{c}{\Skip}{\decomp{R}{c}{c'}{S}} \ \land \\
\valid{\decomp{R}{c}{c'}{S}}{\Skip}{c'}{S}
\end{array}\right)
\]
\end{lemma}
\begin{proof}
  From the definitions of $\valid{R}{c}{c'}{S}$ and $\decomp{R}{c}{c'}{S}$.
\end{proof}

Now $\valid{R}{c}{\Skip}{S}$ talks about only one execution (since the
second command~$\Skip$ doesn't execute by definition).
And indeed it can be expressed equivalently as a statement in
under-approximate Hoare logic.

\newcommand{\validH}[3]{\langle#1\rangle\ #2\ \langle#3\rangle}
For Hoare logic pre- and post-conditions~$P$ and~$Q$ and command
$c$, we write $\validH{P}{c}{Q}$ to mean that~$Q$ \emph{under-approximates}
the post-states~$t$ that $c$ can reach beginning from
a~$P$-pre-state~$s$~\cite{deVries_Koutavas_11,OHearn_19}.

\newcommand{\pred}[2]{#1(#2)}

\begin{definition}[Under-approximate Hoare logic validity]\label{defn:validH}
\[
\validH{P}{c}{Q} \equiv (\forall\ t.\ \pred{Q}{t} \implies (\exists\ s.\ \pred{P}{s} \land \sem{c}{s}{t}))
\]
\end{definition}
Then the following result is a trivial consequence:
\begin{lemma}\label{lem:blah}
  \[
  \valid{R}{\Skip}{c'}{Q} \iff (\forall\ t.\ \validH{\pred{R}{t}}{c'}{\pred{S}{t}})
  \]
\end{lemma}

\cref{thm:decomp} below follows straightforwardly from \cref{lem:decomp}
and~\cref{lem:blah}. It  states how every under-approximate
relational statement
$\valid{R}{c}{c'}{S}$ can be decomposed into two under-approximate
Hoare logic statements: one about~$c$ and the other about~$c'$.

\begin{theorem}{Under-Approximate Relational Decomposition}\label{thm:decomp}
  
  $\valid{R}{c}{c'}{S}$ if and only if
  \[
\left(\begin{array}{l}
  \forall\ t'.\ \validH{(\lambda\ t.\ \rel{R}{t}{t'})}{c}{(\lambda\ t.\ \rel{\decomp{R}{c}{c'}{S}}{t}{t'})}\  \land\\
  \forall\ t.\ \validH{\pred{\decomp{R}{c}{c'}{S}}{t}}{c'}{\pred{Q}{t}}
  \end{array}\right)
\]
\end{theorem}

This theorem implies that, just as over-approximate Hoare logic can be used
to prove over-approximate relational properties~\cite{Beringer_11},
\emph{all} valid under-approximate relational properties can be proved
using (sound and complete) under-approximate Hoare logic reasoning.
This theorem therefore provides
strong evidence for the applicability of under-approximate reasoning methods
(such as under-approximate symbolic execution) for proving under-approximate
relational properties, including e.g.\  via product program
constructions~\cite{barthe2011secure,barthe2011relational,barthe2013beyond,Eilers2018}.

\section{Using the Logic}
\label{sec:using}

While the logic is complete, its rules on their own are inconvenient
for reasoning about common relational properties.
Indeed, often when reasoning in relational logics one is reasoning about
two very similar programs~$c$ and~$c'$. For instance, when proving
noninterference~\cite{Goguen_Meseguer_82} $c$ and~$c'$ might be identical, differing only
when the program branches on secret data~\cite{Murray_MBGK_12}. When proving refinement from
an abstract program~$c$ to a concrete program~$c'$, the two may often
have similar structure~\cite{Cock_KS_08}. 
Hence for the logic to be usable one can derive familiar-looking proof
rules that talk about similar programs.

\subsection{Deriving Proof Rules via Soundness and Completeness}

For instance, the following rule is typical of relational logics
(both over- and under-approximate) for decomposing matched sequential
composition. 

\begin{mathpar}
  \infer{\prove{R}{c}{c'}{S} \and \prove{S}{d}{d'}{T}}
        {\prove{R}{(\Seq{c}{d})}{(\Seq{c'}{d'})}{T}}
        \textsc{Seq-Matched}
\end{mathpar}

Rather than appealing to the rules of~\cref{fig:rules}, this rule is
simpler to derive from the corresponding property on under-approximate
relational validity. In particular, it is straightforward to prove
from \cref{defn:valid} and the big-step semantics of IMP
that when $\valid{R}{c}{c'}{S}$ and
$\valid{S}{d}{d'}{T}$ hold, then $\valid{R}{(\Seq{c}{d})}{(\Seq{c'}{d'})}{T}$
follows. The proof rule above is then a trivial
consequence of soundness and
completeness.

Other such rules that we have also proved by appealing to corresponding
properties of
under-approximate relational validity with soundness and completeness include
the following one for reasoning about matched while-loops.

\newcommand{\two}[1]{#1^2}

\begin{mathpar}
    \infer{\forall\ n. \prove{R(n) \land \one{b} \land \two{b'}}{c}{c'}{R(n+1)} \\
           \prove{\exists\ n.\ R(n)}{\While{b}{c}}{\While{b'}{c'}}{S}}
          {\prove{R(0)}{\While{b}{c}}{\While{b'}{c'}}{S}}
          \textsc{Back-Var-Matched}          
\end{mathpar}

\subsection{Derived Proof Rules}

Other rules for matched programs
are more easily derived by direct appeal to existing proof rules.
For instance, the following rule for matched assignments can be derived
straightforwardly from \textsc{Assign} and~\textsc{Sym}.
\[
\infer{}{\prove{R}{\Assign{x}{e}}{\Assign{x'}{e'}}{\AssignPost{\AssignPost{R}{x}{e}}{x'}{e'}}}\textsc{Assign-Matched}
\]
Similarly the following rule follows from \textsc{Back-Var-Matched} above,
with \textsc{WhileFalse} and~\textsc{Sym}.
\begin{mathpar}
    \infer{\forall\ n. \prove{R(n) \land \one{b} \land \two{b'}}{c}{c'}{R(n+1)}}
          {\prove{R(0)}{\While{b}{c}}{\While{b'}{c'}}{\exists\ n.\ R(n) \land \one{\lnot b} \land \two{\lnot b'}}}
          \textsc{Back-Var-Matched2} 
\end{mathpar}
It is useful for reasoning past two matched loops (rather than for detecting
violations of relational validity within their bodies) and is an almost
direct relational analogue of the backwards variant rule
from~\cite{OHearn_19}.

\subsection{Examples}


\newcounter{linenumbercounter}
\newlength{\codeindentlength}
\setlength{\codeindentlength}{0.5em}
\newcommand{\increaseindent}{\addtolength{\codeindentlength}{1em}}
\newcommand{\decreaseindent}{\addtolength{\codeindentlength}{-1em}}
\newcommand{\indentcode}{\hspace*{\codeindentlength}}

\newcommand{\printcodelinenumber}{{\scriptsize\arabic{linenumbercounter}}\stepcounter{linenumbercounter}}
\newcommand{\codestart}{\setcounter{linenumbercounter}{1}\printcodelinenumber\,}

\newcommand{\codenewline}{\\\printcodelinenumber\indentcode}

\newenvironment{examplecode}{\codestart}{}

\renewcommand{\Seq}[2]{#1;\codenewline#2}
\renewcommand{\Assign}[2]{#1\ensuremath{{}\mathbin{:=}{}}#2}
\renewcommand{\ITE}[3]{\kw{if}\ #1\increaseindent\codenewline#2\decreaseindent\codenewline\kw{else}\increaseindent\codenewline#3\decreaseindent}
\newcommand{\ITNoElse}[2]{\kw{if}\ #1\increaseindent\codenewline#2\decreaseindent}
\renewcommand{\While}[2]{\kw{while}\ #1\ \kw{do}\increaseindent\codenewline#2\decreaseindent}

\newcommand{\var}[1]{\textsf{#1}}
\newcommand{\low}{\var{low}}
\newcommand{\varx}{\var{x}}
\newcommand{\varn}{\var{n}}
\newcommand{\vary}{\var{y}}
\newcommand{\high}{\var{high}}

We illustrate the logic on a few small examples. Each of these examples
concerns the relational property of noninterference~\cite{Goguen_Meseguer_82}.
Under-approximate logics are aimed to prove the existence of bugs rather
than to prove their absence~\cite{OHearn_19}.  Each example has a security bug that
violates noninterference that we provably demonstrate via our logic.

In a programming language-based setting~\cite{Sabelfeld_Myers_03}, 
\emph{termination-insensitive} noninterference holds for a program~$c$ if,
when we execute it from two initial states~$s$ and~$s'$ that agree on the
value of all public (aka ``low'') variables, then whenever both executions
terminate the resulting states~$t$ and~$t'$ must also agree on the values
of all public variables.

\newcommand{\LowEq}{\textsf{L}}

Taking the variable $\low$ to be the only
public variable, we can define when two states agree on the value of
all public variables, written \LowEq, as follows:
\[
\rel{\LowEq}{s}{s'} \equiv (s(\low) = s'(\low))
\]
Then a program~$c$ is insecure if it has two executions beginning in
states~$s$ and~$s'$ for which $\rel{\LowEq}{s}{s'}$ that terminate
in states~$t$ and~$t'$ for which $\rel{\lnot \LowEq}{t}{t'}$.
This holds precisely when there exists some non-empty relation~$S$, 
for which
\[
S \implies \lnot \LowEq \text{ and } \valid{\LowEq}{c}{c}{S}
\]
In fact, since our logic is sound, if we can prove
$\prove{\LowEq}{c}{c}{S}$ for some satisfiable $S$ such that
$S \implies \lnot \LowEq$, then the program \emph{must} be insecure.

\begin{figure}
  \begin{tabular}{ccc}
    \begin{minipage}{0.3\textwidth}
\begin{examplecode}
\ITE{$\varx > 0$}{\Assign{\low}{1}}{\Assign{\low}{0}}
\end{examplecode}
\end{minipage}
    &
\begin{minipage}{0.3\textwidth}
\begin{examplecode}
\Seq{\Assign{$\varx$}{0}}
{\Seq{\While{$\varn > 0$}
            {\Seq{\Assign{$\varx$}{$\varx + \varn$}}
                 {\Assign{$\varn$}{$\vary$}}
            }
     }
     {\ITE{$\varx = 2000000$}{\Assign{\low}{\high}}{\Skip}}
}  
\end{examplecode}
\end{minipage}
&
\begin{minipage}{0.3\textwidth}
\begin{examplecode}
\Seq{\Assign{$\varx$}{0}}
    {\While{$\varx < 4000000$}
            {\Seq{\Assign{$\varx$}{$\varx + 1$}}
            {\ITE{$\varx = 2000000$}{\Assign{\low}{\high}}{\Skip}}
            }
     }
\end{examplecode}
\end{minipage}
  \end{tabular}
  \caption{Examples (some inspired from~\cite{OHearn_19}).\label{fig:examples}}
\end{figure}

The left program in \cref{fig:examples} is trivially insecure and
this is demonstrated straightforwardly by using a derived proof rule
(omitted for brevity)
for matching if-statements that follows the then-branch of the first
and the else-branch of the second, and then by applying the matched assignment
rule~\textsc{Assign-Matched}.

The middle program of \cref{fig:examples} is also insecure, and is
inspired by the \texttt{client0} example of~\cite{OHearn_19}.
To prove it insecure, we first use the consequence rule to strengthen
the pre-relation to one that ensures that the values of the $\high$
variable in both states are distinct. We use the matched assignment and
sequencing rules, \textsc{Assign-Matched} and~\textsc{Seq-Matched}.
Then we apply some derived
proof rules (omitted) for matched while-loops that (via the
disjunction rule \textsc{Disj}) consider both the case in which both loops
terminate immediately (ensuring $\varx = 0$)
and the case in which both execute just once (ensuring $\varx > 0$). Taking
the disjunction gives the post-relation for the loops that $\varx \geq 0$
in both states,
which is of course the pre-relation for the following if-statement.
The rule of consequence allows this pre-relation to be strengthened
to $\varx = 2000000$ in both states,
from which a derived rule (omitted) for matching
if-statements is applied that explores both then-branches, observing
the insecurity.

The rightmost example of \cref{fig:examples} makes full use of the
rule \textsc{Back-Var-Matched}. This rule is applied first with the
backwards-variant relation~$R(n)$ instantiated to require that in both states
$\varx = n$  and $0 \leq \varx < 2000000$.
This effectively reasons over both loops up to the critical iteration
in which $\varx$ will be incremented to become $2000000$.
At this point, we apply a derived rule for matching while loops
similar to \textsc{WhileTrue} that unfolds both for one iteration.
This allows witnessing the insecurity. We then need to reason over the
remainder of the loops to reach the termination of both executions.
This is done using the \textsc{Back-Var-Matched2} rule, with~$R(n)$
instantiated this time to require that in both states $\varx = 2000000 + n$
and $2000000 \leq \varx \leq 4000000$  and, critically,
that $\low = \high$ in both states (ensuring that the security violation
is not undone during these subsequent iterations).

\section{Related Work}

Over-approximate relational logics have been well-studied. Indeed
recent work purports to generalise many previous and forthcoming relational
logics~\cite{maillard2019next}. Our logic is different from these prior
relational logics in that those were all over-approximate. Ours, on the
other hand, is
an under-approximate logic~\cite{OHearn_19,deVries_Koutavas_11}.
Thus, while traditional logics aim to provably establish when relational
properties hold, ours aims instead for provably demonstrating \emph{violations}
of relational properties.

The power of under-approximate logics for provably detecting bugs was
recently explained by O'Hearn~\cite{OHearn_19}. He extended the
previous Reverse Hoare Logic of de Vries and
Koutavas~\cite{deVries_Koutavas_11}, so that the semantics explicitly tracks
erroneous executions, allowing
post-conditions to distinguish between normal and erroneous execution,
producing an Incorrectness Logic.

The semantics of our language, IMP, like that of de Vries and Koutavas
does not explicitly track erroneous execution.
Doing so offers less benefit for the purposes of
relational reasoning, where erroneous behaviour cannot be characterised
by individual executions but rather only by \emph{comparing} two executions.

Just as our relational logic can be seen as the relational analogue of
prior under-approximate Hoare logics,
our decomposition result (\cref{sec:decomposition}) can also be seen as the
relational analogue of Beringer's decomposition result for over-approximate
relational logic~\cite{Beringer_11}. 

\section{Conclusion}

We presented the first under-approximate relational logic, for the simple
imperative language IMP. Our logic is sound and complete. We also showed
a decomposition principle allowing under-approximate relational logic
assertions to be proved via under-approximate Hoare logic reasoning.
We briefly discussed our logic's application to some small
examples, for provably demonstrating the presence of insecurity.
Being a general relational logic,
  it should be equally applicable for provably demonstrating violations
  of other relational properties too.

O'Hearn's Incorrectness Logic~\cite{OHearn_19} is an under-approximate
Hoare logic that was very recently extended to produce the first
Incorrectness Separation Logic~\cite{raadlocal}.
Over-approximate relational separation logics have been
studied~\cite{Yang07,Ernst_Murray_19} and an interesting direction for future research would
include investigation of under-approximate relational separation logics.
Others naturally include the use of under-approximate Hoare logic reasoning
tools for under-approximate relational verification. We hope that this
paper serves as a step towards automatic, provable detection
of relational incorrectness.

\bibliographystyle{splncs04}
\bibliography{references}

\end{document}